%% file: main.tex
\documentclass[11pt]{article}
\usepackage[margin=1in]{geometry}

\usepackage{times}
\usepackage{amsmath}
\usepackage{algorithm}
\usepackage{amsthm}
\usepackage{graphicx,epsfig}
\usepackage{caption}
\usepackage[noend]{algpseudocode}

\usepackage{tikz}

\theoremstyle{plain} \numberwithin{equation}{section}
\newtheorem{theorem}{Theorem}[section]
\newtheorem{corollary}[theorem]{Corollary}

\newtheorem{lemma}[theorem]{Lemma}

\theoremstyle{plain}
\newtheorem{invariant}[theorem]{Invariant}

\newcommand{\N}{\mathcal{N}}

\newcommand{\C}{\mathcal{C}}

\newcommand{\eps}{\epsilon}

\newcommand{\dur}{dur}

\newcommand{\psdur}{psdur}

\ifdefined\DEBUG

\newcommand{\fab}[1]{\textcolor{red}{#1}}

  \def\rem#1{{\marginpar{\raggedright\scriptsize #1}}}
  \newcommand{\fabr}[1]{\rem{\textcolor{red}{$\bullet$ #1}}}
  \newcommand{\sayr}[1]{\rem{\textcolor{magenta}{$\bullet$ #1}}}
  \newcommand{\janr}[1]{\rem{\textcolor{blue}{$\bullet$ #1}}}
  \newcommand{\shar}[1]{\rem{\textcolor{orange}{$\bullet$ #1}}}
  \newcommand{\quar}[1]{\rem{\textcolor{green}{$\bullet$ #1}}}
\else
  \newcommand{\fab}[1]{#1}

  \newcommand{\fabr}[1]{}
  \newcommand{\sayr}[1]{}
  \newcommand{\janr}[1]{}
  \newcommand{\shar}[1]{}
  \newcommand{\quar}[1]{}
\fi

\def\cN{{\cal N}}

\def\cE{{\cal E}}
\newcommand {\ignore} [1] {}
\def\eps{\epsilon}

\def\MathP{\hbox{\rm I\kern-2pt P}}
\newcommand{\Prob}{\MathP}

\newcommand{\ins}{\texttt{handle-insertion}}
\newcommand{\del}{\texttt{handle-deletion}}

\newcommand{\dout}{d_{\texttt{out}}}

\newcommand{\setlvl}{\texttt{set-level}}
\newcommand{\updateins}{\texttt{update-color-edge-insertion}}
\newcommand{\update}{\texttt{update-color}}

\newcommand{\base}{3}
\newcommand{\cp}[1]{\mathcal{C}^+_{#1}}
\newcommand{\cP}{\mathcal{P}}
\newcommand{\colorarr}[1]{\mathcal{C}_{#1}}
\newcommand{\p}[1]{\mathcal{P}_{#1}}
\newcommand{\cc}[2]{\mu^+_{#1}(#2)}
\newcommand{\recolor}{\texttt{recolor}}
\newcommand{\detc}{\texttt{det-color}}
\newcommand{\randc}{\texttt{rand-color}}

\newcommand{\pp}[1]{p^+_{#1}}
\newcommand{\po}[1]{p_{#1}}
\newcommand{\lbelow}{\phi_v}

\newcommand{\down}[1]{\mathcal{D}_{#1}}
\newcommand{\up}[1]{\mathcal{U}_{#1}}
\newcommand{\epoch}{\mathcal{E}}

\newcommand{\al}[2]{A_{#1}^{#2}}

\begin{document}

\title{Fully Dynamic $(\Delta+1)$-Coloring in Constant Update Time}

\author{Sayan Bhattacharya\thanks{University of Warwick, \texttt{jucse.sayan@gmail.com}} \and Fabrizio Grandoni\thanks{IDSIA, \texttt{fabrizio@idsia.ch}. Partially supported by the SNSF Excellence Grant 200020B$\_$182865/1} \and Janardhan Kulkarni\thanks{Microsoft, \texttt{jakul@microsoft.com}} \and Quanquan C. Liu\thanks{MIT, \texttt{quanquan@mit.edu}} \and Shay Solomon\thanks{Tel Aviv University, \texttt{solo.shay@gmail.com}} \and}

\date{}

\maketitle

\thispagestyle{empty} 

\newcounter{list}

\renewcommand{\L}{\mathcal{L}}

\begin{abstract}
\noindent 
The problem of (vertex) $(\Delta+1)$-coloring a graph of maximum degree $\Delta$ has been extremely well-studied over the years in various settings and models.
Surprisingly, for the dynamic setting, almost nothing was known until recently.
In SODA'18, Bhattacharya, Chakrabarty, Henzinger and Nanongkai devised a randomized data structure for maintaining a $(\Delta+1)$-coloring
with $O(\log \Delta)$ expected amortized update time. In this paper, we present a $(\Delta+1)$-coloring data structure that achieves a constant amortized update time and show that this time bound holds not only in expectation but also with high probability.
\footnote{An earlier version of this paper started to circulate in early July 2019.}
\end{abstract}

\newpage

\setcounter{page}{1}
\pagenumbering{arabic}

\section{Introduction}

Vertex coloring is one of the most fundamental and best studied graph problems. Consider any integral parameter $\lambda > 0$, an undirected graph $G = (V, E)$ with $n$ nodes and $m$ edges, and  a {\em palette} $\C = \{1, \ldots, \lambda\}$ of $\lambda$ colors. A {\em $\lambda$-coloring} in $G$ is simply a function $\chi : V \rightarrow \C$ which assigns a color $\chi(v) \in \C$ to each vertex $v \in V$. Such a coloring is called {\em proper} iff no two neighboring nodes in $G$ get the same color. The main goal is to compute a proper $\lambda$-coloring in the input graph $G = (V, E)$ such that $\lambda$ is as small as possible. Unfortunately this problem is NP-hard and even extremely hard to approximate: for any constant $\eps>0$, there is no polynomial-time approximation algorithm with approximation factor $n^{1-\eps}$ unless $P\neq NP$ \cite{FK98,KP06,Z07}. In contrast, there is a textbook greedy algorithm that runs in $O(m+n)$ time and computes a $(\Delta+1)$-coloring when $\Delta$ is an upper bound on the maximum degree of the input graph $G = (V, E)$.

We address the problem of maintaining a proper $(\Delta+1)$-coloring in the \emph{fully dynamic} setting. Here, the input graph $G = (V, E)$ changes via a sequence of {\em updates}, where each update consists of the insertion or deletion of an edge in $G$. There is a {\em fixed} parameter $\Delta > 0$ such that the maximum degree in $G$ remains upper bounded by $\Delta$ throughout this update sequence. We want to design a data structure that is capable of maintaining a proper $(\Delta+1)$-coloring in such a dynamic graph $G$. The time taken by the data structure to handle an update is called its {\em update time}. We say that a data structure has an {\em amortized} update time of $O(\gamma)$ iff starting from an empty graph, it takes at most $O(t \cdot \gamma)$ time to handle any sequence of $t$ updates. Our goal is to ensure that the update time of our data structure is as small as possible.

There is a naive data structure for this problem that has $O(\Delta)$ update time, which works as follows. Suppose that we are maintaining a proper $\Delta+1$-coloring $\chi : V \rightarrow \C$ in $G$. At this point, if an edge gets deleted from the graph, then we do nothing, as the coloring $\chi$ continues to remain proper. Otherwise, if an edge $uv$ gets inserted into $G$, then we first check if $\chi(u) = \chi(v)$. If not, we do nothing. If yes, then we pick an arbitrary endpoint $x \in \{u, v\}$, and by scanning all its neighbors we identify a {\em blank} color $c' \in \C$ for $x$ (one that is not assigned to any of its neighbors). Such a blank color is guaranteed to exist, since $x$ has at most $\Delta$ neighbors and the palette $\C$ consists of $\Delta+1$ colors. We now {\em recolor} the node $x$ by assigning it the color $c'$. This  results in a proper $(\Delta+1)$-coloring in the current graph. The time taken to implement this procedure is proportional to the degree of $x$, hence it is at most $O(\Delta)$. 

It is natural to ask if we can beat  this naive data structure. In particular, can we design a data structure for this problem that has {\em polylogarithmic} update time? In SODA'18, this question was answered affirmatively  by Bhattacharya, Chakrabarty, Henzinger and Nanongkai~\cite{BCHN18}, who obtained the following result.

\begin{theorem}
\label{th:old}\cite{BCHN18}
There is a randomized data structure that can maintain a $\Delta+1$-coloring in a dynamic graph with $O(\log \Delta)$ amortized update time in expectation.
\end{theorem}

Our main result  is summarized in Theorem~\ref{thr:main} below. We design a randomized data structure for $(\Delta+1)$-coloring that has $O(1)$ amortized update time in expectation and  with high probability (for a sufficiently long update sequence). This constitutes a {\em dramatic} improvement over the update time of~\cite{BCHN18} as stated in Theorem~\ref{th:old}. As is the case with most existing data structures that are randomized, both Theorems~\ref{th:old} and~\ref{thr:main} hold only when the adversary deciding the next update is {\em oblivious} to the past random choices made by the data structure. 

\begin{theorem}\label{thr:main}
There is a randomized data structure for maintaining a $(\Delta+1)$-coloring  in a dynamic graph that, given any sequence of $t$ updates, takes total time $O(n\log n+n\Delta+t)$ in expectation and with high probability. The space usage is $O(n\log n+n\Delta+m)$, where $m$ is the maximum number of edges present at any time.
\end{theorem}
We remark that, unlike several related results in the literature, our amortized bound holds also with high probability.

\subsection{Our Technique}
\label{sec:technique}

We start with a high level overview of the data structure in~\cite{BCHN18}. Roughly speaking, they maintain a {\em hierarchical partition} of the node-set $V$ into $O(\log \Delta)$ levels. Let $\ell(v) \in \{1, \ldots, \log \Delta\}$ denote the {\em level} of a node $v \in V$. For every edge $uv \in E$, say that $u$ is a {\em same-level-neighbor}, {\em down-neighbor} and  {\em up-neighbor} of $v$ respectively iff $\ell(u) = \ell(v)$, $\ell(u) < \ell(v)$ and $\ell(u) \geq \ell(v)$. The  following invariant is maintained. 

\begin{invariant}
\label{inv:main:old}
Each node $v \in V$ has  $\Omega(2^{\ell(v)})$ down-neighbors and  $O(2^{\ell(v)})$ same-level neighbors.
\end{invariant}

In order to ensure that Invariant~\ref{inv:main:old} holds, the nodes need to keep changing their levels as the input graph keeps getting updated via a sequence of edge insertions/deletions.   It is important to note that the subroutine  in charge of maintaining this invariant is  {\em deterministic} and  has $O(\log \Delta)$ amortized update time. 

The data structure in~\cite{BCHN18} uses a separate (randomized) subroutine to maintain a proper $(\Delta+1)$-coloring in the input graph, on top of the hierarchical partition. To appreciate the main intuition behind this {\em recoloring subroutine}, consider the insertion of an edge $uv$ at some time-step $\tau$, and suppose that both $u$ and $v$ had the same color  just before this insertion.  Pick any arbitrary endpoint $x \in \{u, v\}$. The data structure picks a new color for $x$ as follows. Let $\C_x \subseteq \C$ denote the subset of colors that satisfy the following property at time-step $\tau$: A color $c \in \C$ belongs to $\C_x$ iff either (a) no neighbor of $x$ has color $c$, or (b) no up-neighbor of $x$ has color $c$ and exactly one down-neighbor of $x$ has color $c$. Since the node $x$ has at most $\Delta$ neighbors and the palette $\C$ consists of $\Delta+1$ colors, a simple counting argument (see the proof of Lemma~\ref{lem:inv:recolor}) along with Invariant~\ref{inv:main:old} implies that the size of the set $\C_x$ is at least $\Omega(2^{\ell(x)})$.   Furthermore, using appropriate data structures, the set $\C_x$ can be computed in time proportional to the number of down-neighbors and same-level neighbors of $x$, which is at most $O(2^{\ell(x)})$ by Invariant~\ref{inv:main:old}.
The data structure picks a color $c'$ uniformly at random from the set $\C_x$, and then recolors $x$ by assigning it the color $c'$. By definition of the set $\C_x$, at most  one neighbor (say, $y$) of $x$ has the color $c'$, and, furthermore, if such a neighbor $y$ exists then $\ell(y) < \ell(x)$. If the down-neighbor $y$ exists, then we recursively recolor $y$ in the same manner. Note that this entire procedure leads to a {\em chain} of recolorings. However, the levels of the nodes involved in these successive recolorings form a strictly decreasing sequence. Thus, the total time taken by the subroutine to handle the edge insertion is at most $\sum_{\ell=1}^{\ell(x)} O(2^{\ell}) = O(2^{\ell(x)})$. 

Now comes the most crucial observation. Note that each time the data structure recolors a node $x$, it picks a new color uniformly at random from a set of size $\Omega(2^{\ell(x)})$. Thus, intuitively, if the adversary deciding the update sequence is oblivious to the random choices made by the data structure, then in expectation at least $\Omega(2^{\ell(x)}/2) = \Omega(2^{\ell(x)})$ edge insertions incident on $x$ should take place before we encounter a {\em bad event} (where the other endpoint of the edge being inserted has the same color as $x$). The discussion in the preceding paragraph implies that we need $O(2^{\ell(x)})$ time to handle the bad event. Thus, overall we get an amortized update time of $O(1)$ in expectation. 

\smallskip
\noindent {\bf Our contribution:} To summarize, the data structure in~\cite{BCHN18} has two components -- (1) a deterministic subroutine for maintaining the hierarchical partition which takes $O(\log \Delta)$ amortized update time, and (2) a randomized subroutine for maintaining a proper $(\Delta+1)$-coloring  which takes $O(1)$ amortized update time. The analysis of the amortized update time of the first subroutine is done via an intricate potential function, and it is not clear if it is possible to improve the update time of this subroutine to $O(1)$. 

In order to get an overall update time of $O(1)$, our data structure merges these two components together in a very careful manner. Our starting point is to build on the high-level strategy used for maximal matching in~\cite{S16focs}, but due to the differences between the two problems our argument deviates from that of~\cite{S16focs} significantly in several crucial and highly nontrivial points.
 Suppose that we decide to recolor a node $x$ during the course of our data structure (either due to the insertion of an edge incident on it, or because one of its up-neighbors took up the same color as $x$ while recoloring itself). Let $\ell(x)$ be the current level of $x$. We first check if the number of down-neighbors of $x$ is $\Omega(3^{\ell(x)})$. If the answer is yes, then we move up the node $x$ to the minimum level $\ell'(x) > \ell(x)$ where the number of its down-neighbors becomes $\Theta(3^{\ell'(x)})$, following which we recolor the node $x$ in the same manner as  in~\cite{BCHN18}. In contrast, if the answer is no, then we find a new color for $x$ that does not conflict with any of its neighbors and move the node $x$ down to the smallest level in the hierarchical partition. Thus, in our data structure, the hierarchical partition itself is determined by the random choices made by the nodes while they recolor themselves. This  makes the analysis of our randomized data structure significantly more challenging, and Invariant~\ref{inv:main:old} is no longer satisfied all the time. Nevertheless, we manage to show that our new data structure has $O(1)$ amortized update time not only in expectation, but also with high probability.
 
 \medskip
\noindent {\bf Independent work:} Independently of our work, Henzinger and Peng~\cite{HenzingerP19} have obtained a data structure for $(\Delta+1)$-vertex coloring with $O(1)$ \emph{expected amortized update time}. Note that our work achieves $(\Delta + 1)$-vertex coloring with $O(1)$ amortized update time \emph{with high probability}.
 
\subsection{Perspective}

$(\Delta+1)$-Vertex coloring of a graph is a {\em local problem} where we are asked to assign a {\em state}  to every node/edge, subject to a {\em local constraint} at each node/edge that specifies the admissible assignments of states  in its immediate neighborhood. For instance, in a $(\Delta+1)$-vertex coloring the state of a node corresponds to its color, and the local constraint at a node requires that none of its neighbors gets the same color as the node itself. {\em Maximal matching}, {\em maximal independent set} (MIS) and {\em $(2\Delta-1)$-edge coloring} are some of the other well-known problems that fall within this category. Interestingly, all these problem admit the same type of  greedy algorithm in the static setting: Scan  the nodes/edges in any arbitrary order. While considering a given node/edge during this scan, assign a state to it depending on the states already assigned to its preceding neighbors. This greedy algorithm runs is linear time. Since this is the best possible running time in the static setting,  it is natural to ask if  we can match the performance of this greedy algorithm when the input graph changes dynamically, which leads us to the question of designing data structures for all these local problems in polylogarithmic (and, in an ideal scenario -- constant) update time. This research agenda has received significant attention in the dynamic algorithms community in recent years.
Baswana et al.~\cite{BGS15} gave a randomized data structure for dynamic maximal matching with $O(\log n)$ amortized update time, which was later improved to $O(1)$ by Solomon~\cite{S16focs}. For $(2\Delta-1)$-edge coloring, Bhattacharya et al.~\cite{BCHN18} gave a deterministic data structure with $O(\log \Delta)$ worst case update time (see~\cite{DHZ19} for another related result). Finally, there has been a spate of results on dynamic MIS~\cite{AssadiOSS18,OnakSSW18,AssadiOSS19}, culminating in two very recent papers~\cite{MISfocs1,MISfocs2} that achieve polylogarithmic update time.

To appreciate why getting $O(1)$ update times for natural dynamic problems is an important research agenda, note that there is a very influential series of results on  giving {\em cell-probe} lower bounds for dynamic problems. For quite a few well-known dynamic problems, these lower bounds rule out the possibility of getting $o(\log n)$ update time~\cite{LarsenWY18,Larsen12,PatrascuD06}. To take a concrete example, consider the dynamic connectivity problem, where the input graph changes via a sequence of edge insertions/deletions, and we have to answer queries of the form $Q(u, v)$ which asks whether or not the two nodes $u$ and $v$ belong to the same connected component. Patrascu and Demain~\cite{PatrascuD06} showed that any dynamic connectivity data structure must either have an amortized update time of $\Omega(\log n)$ or a query time of $\Omega(\log n)$. In contrast, by getting an $O(1)$ update time data structure, we rule out the possibility of such a lower bound for dynamic $(\Delta+1)$-coloring. 

To summarize,  our result fits nicely within a  line of work  where the focus is to design data structures with $O(1)$ update times. We conclude by mentioning a couple of  recent results on this topic. For every fixed $\epsilon > 0$, it was shown that one can deterministically maintain in $O(1)$ amortized update times a $(2+\epsilon)$-approximate minimum vertex cover in a fully dynamic graph~\cite{BhattacharyaK19}, and  a $(1+\epsilon)$-approximate maximum matching in a graph undergoing a sequence of edge insertions in $O(1)$ update time~\cite{GLSSS19soda}.

\section{Our Algorithm}\label{sec:algo}

Consider a graph $G = (V, E)$ with $|V| = n$ nodes that is changing via a sequence of {\em updates} (edge insertions and deletions). Let $\Delta > 0$ be a fixed integer such that the maximum degree of any node in the dynamic graph $G$ is always upper bounded by $\Delta$. Let $\C = \{1, \ldots, \Delta+1\}$ denote a palette of $\Delta+1$ colors. Our algorithm will maintain a proper $\Delta+1$-coloring $\chi : V \rightarrow \C$ in the dynamic graph $G$.

\medskip
\noindent {\bf A hierarchical partition of the node-set $V$:} Fix a parameter $L = \lceil \log_3 (n-1) \rceil - 1$. Our dynamic algorithm will maintain a hierarchical partition of the node-set $V$ into $L+2$ distinct {\em levels} $\{-1, 0, \ldots, L\}$. We let $\ell(v) \in \{-1, 0, \ldots, L\}$ denote the level of a given node $v \in V$. The levels of the nodes will vary over time. Consider any edge $uv \in E$ in the dynamic graph $G$ at any given point in time: We say that $u$ is an {\em up-neighbor} of $v$ iff $\ell(u) \geq \ell(v)$, and a {\em down-neighbor} of $v$ iff $\ell(u) < \ell(v)$.

\medskip
\noindent {\bf Notations:} We now introduce a few important notations. Fix any node $v \in V$. Let $\N_v = \{ u \in V : uv \in E \}$ denote the set of neighbors of $v$. Furthermore, let $\C^+_v = \{ c \in \C : c = \chi(u) \text{ for some } u \in \N_v \text{ with } \ell(u) \geq \ell(v)\}$ denote the set of colors assigned to the up-neighbors of $v$. We say that $c \in \C$ is a {\em blank} color for $v$ iff no neighbor  of $v$ currently has the color $c$. Similarly, we say that $c \in \C$ is a {\em unique} color for $v$ iff $c \notin \C^+_v$ {\em and} exactly one down-neighbor of $v$ currently has the color $c$. Finally, for every  $k \in \{-1, \ldots, L\}$, we let $\phi_v(\ell) = \left| \{ u \in \N_v : \ell(u) < \ell \} \right|$ denote the number of neighbors of $v$ that currently lie below level $k$.

\medskip
We are now ready to describe our dynamic algorithm.

\medskip
\noindent {\bf Preprocessing:} In the beginning, the input graph $G = (V, E)$ has an empty edge-set, i.e., $E = \emptyset$, and the algorithm starts with any arbitrary coloring $\chi : V \rightarrow \C$. All the relevant data structures are initialized. Subsequently, the algorithm handles the sequence of updates to the input graph in the following manner.

\medskip
\noindent {\bf Handling the deletion of an edge:} Suppose that an edge $uv$ gets deleted from $G$. Just before this deletion, the coloring $\chi : V \rightarrow \C$ maintained by the algorithm was proper (no two adjacent nodes had the same color). So the coloring $\chi$ continues to remain proper even after the deletion of the edge $uv$. Accordingly, the deletion of an edge does {\em not} lead to any change in the levels of the nodes and the coloring maintained by the algorithm.

\medskip
\noindent {\bf Handling the insertion of an edge:} This procedure is described in Figure \ref{fig:insertion}. Suppose that an edge $uv$ gets inserted into $G$. 
If, just before this insertion, we had $\chi(u) \neq \chi(v)$, then we call this insertion \emph{conflict-less}, and otherwise \emph{conflicting}. In case of a conflict-less insertion, the coloring $\chi$ continues to remain proper even after insertion of the edge $uv$. In this case, the edge-insertion does {\em not} lead to any change in the levels of the nodes or the colors assigned to them.  Otherwise, we pick the endpoint $x \in \{u, v\}$ that was recolored last and call the subroutine $\recolor(x)$. 
 This call to $\recolor(x)$ changes the color assigned to $x$ and it might also change the level of $x$. However, there is a possibility that the new color assigned to $x$ might be the same as the color of (at most one) down-neighbor of $x$. If this happens to be the case, then we go to that neighbor of $x$ it conflicts with, and keep repeating the same process until we end up with a proper coloring in the input graph $G$.

Procedure $\recolor(x)$ (see Figure \ref{fig:recolor}), depending on whether $\phi_x(\ell(x)+1) < 3^{\ell(x)+2}$ or not, calls one of the procedures $\detc(x)$ and $\randc(x)$ which are described next.

 \medskip
 \noindent {\bf $\detc(x)$:} This subroutine first picks a {\em blank} color (say) $c$ for the node $x$. Recall that by definition no neighbor of $x$ has the color $c$. It now recolors the node $x$ by setting $\chi(x) \leftarrow c$. Finally, it moves the node $x$ down to level $-1$, by setting $\ell(x) \leftarrow -1$. It then updates all the relevant data structures.

  \medskip
 \noindent {\bf $\randc(x)$:} This subroutine works as follows. Let $\ell = \ell(x)$ be the level of the node $x$ when this subroutine is called. Step 04 in Figure~\ref{fig:recolor} implies that at that time we have $\phi_x(\ell+1) \geq 3^{\ell+2}$. It identifies the {\em minimum} level $\ell' > \ell$ where $\phi_x(\ell'+1) < 3^{\ell'+2}$. Such a level $\ell'$ must exist because  $\phi_x(L+1) \leq (n-1) < 3^{L+2}$. The subroutine then moves the node $x$ up to level $\ell'$, by setting $\ell(x) \leftarrow \ell'$, and updates all the relevant data structures. After this step, the subroutine computes the set $\C^*_x \subseteq \C$ of colors that are either blank or unique for $x$, next called \emph{palette}. It picks a color $c \in \C^*_x$ uniformly at random, and recolors the node $x$ with color $c$, by setting $\chi(x) \leftarrow c$. It then updates all the relevant data structures. If $c$ happens to be a blank color for $x$, then no neighbor of $x$ has the same color as $c$. In other words, this recoloring of $x$ does {\em not} lead to any new conflict. Accordingly, in this case the subroutine returns NULL. Otherwise, if $c$ happens to be a blank color for $x$, then by definition exactly one down\fab{-}neighbor (and zero up\fab{-}neighbors) of $x$ also has color $c$. Let this down\fab{-}neighbor be $y$. In other words, the recoloring of $x$ creates a new {\em conflict} along the edge $(x, y)$, and we need to recolor $y$ to ensure a proper coloring. Thus, in this case the subroutine returns the node $y$.

 \begin{figure}[htbp]
                                                \centerline{\framebox{
                                                                \begin{minipage}{5.5in}
                                                                        \begin{tabbing}                                                                            
                                                                                01. \ \ \=  {\sc If} $\chi(u) = \chi(v)$, {\sc Then} \\
                                                                                02. \> \ \ \ \ \ \ \ \ \= Let $x \in \{u, v\}$ be the endpoint that was recolored last. \\
                                                                                03. \> \> {\sc While} $x \neq \text{NULL}$: \\            
                                                                                04. \> \> \ \ \ \ \ \  \ \ \ \ \= $x \leftarrow \recolor(x)$.                                                                                                                 
                                                                        \end{tabbing}
                                                                \end{minipage}
                                                        }}
                                                        \caption{\label{fig:insertion} Handling the insertion of an edge $uv$.}
                                                \end{figure}

\begin{figure}[htbp]
                                                \centerline{\framebox{
                                                                \begin{minipage}{5.5in}
                                                                        \begin{tabbing}                                                                            
                                                                                01. \ \ \=  {\sc If} $\phi_x(\ell(x)+1) < 3^{\ell(x)+2}$, {\sc Then} \\
                                                                                02. \> \ \ \ \ \ \ \ \ \= $\detc(x)$. \\
                                                                                03. \> \> {\sc Return} NULL. \\
                                                                                04. \> {\sc Else} : \\            
                                                                                05. \> \> \ \ \ \ \ \  \ \ \ \ \= $y \leftarrow \randc(x)$.   \\
                                                                                06. \> \> \> {\sc Return} $y$.                                                                                                              
                                                                        \end{tabbing}
                                                                \end{minipage}
                                                        }}
                                                        \caption{\label{fig:recolor} $\recolor(x)$.}
                                                \end{figure}

It is not hard to implement the above data structure such that the following result holds (more details in appendix).
\begin{lemma}\label{lem:runningTime}
There is an implementation of the above data structure such that:
\begin{enumerate}\itemsep0pt
\item The preprocessing time is $O(n\log n+\Delta n)$;\label{lem:runningTime:prep}
\item The space usage is $O(n\log n+\Delta n+m)$, where $m$ is the maximum number of edges present at any time;\label{lem:runningTime:space}
\item Each deletion and conflict-less insertion takes $O(1)$ time deterministically;\label{lem:runningTime:easy}
\item Procedure $\detc(x)$ takes time $O(3^{\ell(x)})$;\label{lem:runningTime:det}
\item Procedure $\randc(x)$ takes time $O(3^{\ell'(x)})$ where $\ell'(x)>\ell(x)$ is the new level of node $x$ at the end of the procedure.\label{lem:runningTime:rand} 
\end{enumerate}
\end{lemma}
The $O(\Delta n)$ term in the preprocessing time and space usage is obtained by storing informations about the colors of certain neighbors of each node $v$ in arrays of size $\Delta+1$. With a careful (but standard) use of dynamic hash functions, we can make the size of each such array proportional to the number of neighbors of $v$, while keeping constant amortized access time. This way we can remove the term $O(\Delta n)$ without increasing asymptotically the final amortized cost of our data structure.

\section{Analysis}\label{sec:mainanal} 

In this section, we analyze the runtime of our edge update algorithm. We assume that our graph is empty at the end, meaning no edges exist on the graph after we perform all the updates in our update sequence. To ensure we end with an empty graph, we append additional edge deletions at the end of the original update sequence. Since we begin with an empty graph, this at most doubles the number of updates in our update sequence, but simplifies our analysis. Because edge deletions will never cause a recoloring of any vertex and the number of updates increases by at most a factor of $2$, an amortized runtime bound of our algorithm with respect to the new update sequence will imply the same (up to a factor of 2) amortized bound with respect to the original sequence. 

It is not hard to show that our data structure maintains the following invariant.
\begin{invariant}\label{main:inv:recolor} 
Consider a vertex $v$ at level $\ell(v)\geq 0$ \fab{at a given point of time $\tau$. When $v$ was recolored last prior to $\tau$, it chose a color uniformly at random from a} palette of size at least $3^{\ell(v) + 1}/2 + 1$. Furthermore, at that time $v$ has at least $3^{\ell(v)+1}$ down-neighbors. For $\ell(v)=-1$, the color of $v$ is set deterministically.
\end{invariant}
\begin{lemma}
\label{lem:inv:recolor}
Invariant \ref{main:inv:recolor} holds for all vertices at the beginning of each update.\end{lemma}
\begin{proof}
During the preprocessing step the color of each node $v$ is set deterministically to some arbitrary color and $\ell(v)=-1$. Hence the claim holds initially. The color of $v$ changes only due to a call to $\recolor(v)$. Let $\ell(v)$ and $\ell'(v)$ denote the level of $v$ at the beginning and end of this call. 
If $\recolor(v)$ calls $\detc(v)$, the color of $v$ is set determinstically and $\ell'(v)=-1$. Hence the invariant holds. Otherwise, $\recolor(v)$ invokes $\randc(v)$. 
The latter procedure sets $\ell'(v)$ to the smallest value (larger than $\ell(v)$) such that $\phi_v(\ell'(v)+1) < 3^{\ell'(v)+2}$. 
Recall that $\phi_v(\ell)$ is the number of neighbors of $v$ of level smaller than $\ell$. This implies that the number of down-neighbors of $v$ (at level $\ell'(v)$) are  $\phi:=\phi_v(\ell'(v))\geq 3^{\ell'(v)+1}$. 

It is then sufficient to argue that the palette used by $\randc(v)$ has size at least $\phi/2+1$. To this aim we use exactly the same argument as in \cite{BCHN18} (that we reproduce here for the sake of completeness). One has $|\colorarr{v}| = (\Delta + 1) - |\cp{v}| \geq (\Delta + 1) - (\Delta - \phi)=\phi+1$, where equality holds when up-neighbors of $v$ all have distinct colors. 
For any color in $c \in \colorarr{v}$ that is occupied by at most one down-neighbor of $v$, $c$ is a blank or unique color. Let $x$ be the number of down-neighbors of $v$ that occupy a unique color. Then, the size of $v$'s palette is at least $1 + |\phi| - (\phi - x)/2 \geq \phi/2+1$.
\end{proof}

Let $t$ be the total number of updates. Excluding the preprocessing time, the running time of our data structure is given by the cost of handling insertions and deletions. By Lemma \ref{lem:runningTime}-\ref{lem:runningTime:easy}, the total cost of deletions and insertions that do not cause conflicts is $O(t)$. We can therefore focus on insertions that cause conflicts. Modulo constant factors, the total cost of the latter insertions is bounded by the total cost of the calls to $\recolor(\cdot)$ that they  induce. In order to bound this cost, we introduce the notion of epoch in next section.

\subsection{Epochs}

From the previous discussion, we need to bound the total cost of the calls to $\recolor(\cdot)$. To that aim, and inspired by \cite{BGS15}, we introduce the following notion of epoch. An epoch $\epoch$ is associated with a node $v=v(\epoch)$, and consists of any maximal time interval in which $v$ keeps the same color. So $\epoch$ starts with a call to $\recolor(v)$, and ends immediately before the next call to $\recolor(v)$ is executed. Observe that there are potentially multiple epochs associated with the same node $v$. Notice that by construction, during an epoch $\epoch$ the level and color of $v(\epoch)$ does not change: we refer to that level and color as $\ell(\epoch)$ and $\chi(\epoch)$, resp. By $\epoch_\ell$ we denote the set of epochs at level $\ell$. 
We define the \emph{cost} $c(\epoch)$ of an epoch $\epoch$ as the time spent by the call to $\recolor(v(\epoch))$ that starts it, and then we charge the cost of every epoch $\epoch$ at level $\ell(\epoch) = - 1$ to the previous epoch involving the same node $v(\epoch)$. After implementing this charging scheme, it follows from 
Lemma \ref{lem:runningTime} (points \ref{lem:runningTime:det}-\ref{lem:runningTime:rand}) that the cost of any epoch $\epoch$ is given by $c(\epoch) = O\left(3^{\ell(\epoch)} \right)$.

\begin{lemma}
Excluding the preprocessing time, the total running time of the data structure is given by: $O(\sum_{\ell}\sum_{\epoch\in \epoch_\ell}c(\epoch))=O(\sum_\ell |\epoch_\ell|\cdot 3^{\ell(\epoch)})$. 
\end{lemma}

\medskip
\noindent {\bf A classification of epochs:} For our analysis, it will be convenient to classify epochs as follows. An epoch $\epoch$ is \emph{final} if it is not concluded by a call to $\recolor(v(\epoch))$. In other terms, for an final epoch $\epoch$,  $v(\epoch)$ keeps the same color from the beginning of $\epoch$ till the end of all the updates. Otherwise $\epoch$ is \emph{terminated}. A terminated epoch $\epoch$, $v=v(\epoch)$, terminates for two possible events that happen after its beginning: (1) some edge $uv$ is inserted, with $\chi(u)=\chi(v)$, hence leading to a call to $\recolor(v)$; (2) a call to $\recolor(w)$ for some up-neighbor $w$ of $v$ forces a call to $\recolor(v)$  (without the insertion of any edge incident to $v$). We call the epochs of the first and second type \emph{original} and \emph{induced}, resp. In the second case, we say that the epoch $\epoch'$ that starts with the recoloring of $w$ \emph{induces} $\epoch$. 

We now prove a couple of  lemmas that respectively bound the total cost of the induced and final epochs.
\begin{lemma}\label{lem:inducedToOriginal}
The total cost of induced epochs is (deterministically) at most $O(1)$ times the total cost of original and final epochs.
\end{lemma} 
\begin{proof}
Let us construct a directed \emph{epoch graph}, with node set the set of epochs, and a directed edge $(\epoch,\epoch')$ iff $\epoch'$  induced $\epoch$. Notice that, for any edge $(\epoch,\epoch')$ in the epoch graph, $\ell(epoch')>\ell(epoch)$. Observe also that this graph consists of a collection of directed paths ending at original and final epochs. Let us charge the cost of each induced epoch $\epoch$ to the root $r(\epoch)$ of the corresponding path in the epoch graph. All the cost is charged to original and final epochs, and the cost charged to one epoch $\epoch$ of the latter type is at most $\sum_{\ell< \ell(\epoch)}O(3^{\ell})=O(3^{\ell(\epoch)})$. The claim follows.
\end{proof}

\begin{lemma}\label{lem:final}
Given any sequence of $t$ updates, the total cost of final epochs is (deterministically)  $O(t)$.
\end{lemma} 
\begin{proof}
By Invariant \ref{main:inv:recolor}, for any epoch $\epoch$, $v=v(\epoch)$ and $\ell=\ell(\epoch)$, $v$ must have at least $3^{\ell+1}$ down-neighbors at the beginning of $\epoch$. Since by assumption at the end of the process the graph is empty, there must be at least $3^{\ell+1}$ deletions with one endpoint being $v$ during $\epoch$. By charging the $O(3^\ell)$ cost of $\epoch$ to the latter deletions, and considering that each deletion is charged at most twice, we achieve a average cost per deletion in $O(1)$, hence a total cost in $O(t)$. 
\end{proof}

\medskip
\noindent {\bf A classification of levels:} Recall that $\epoch_\ell$ denotes the set of epochs at level $\ell$. 
We now classify the levels into 3 types, as defined below.
\begin{itemize}
\item We say that a level $\ell$ is {\em induced-heavy} iff at least $1/2$-fraction of the epochs in $\epoch_{\ell}$ are induced.
\item We say that a level $\ell$ is {\em final-heavy} iff (a) it is not induced-heavy {\em and} (b) at least $1/8$-fraction of the epochs in $\epoch_{\ell}$ are final. 
\item We say that a level $\ell$ is {\em original-heavy} iff it is neither induced-heavy nor final-heavy. Note that if a level $\ell$ is original-heavy, then at least $3/8$-fraction of the epochs in $\epoch_\ell$ are original. 
\end{itemize} 
Throughout the rest of the paper, we say that an epoch is induced-heavy, final-heavy and original-heavy if it respectively belongs to an induced-heavy, final-heavy and original-heavy level.  Furthermore, we use the term ``{\em cost of a level $\ell$}'' to refer to the total cost of all the epochs at level $\ell$. We next bound the total cost of all the  induced-heavy and final-heavy levels. 

\begin{lemma}\label{lem:lightToHeavy}
The total cost of all the induced-heavy levels is (deterministically) at most $O(1)$ times the total cost of all the original-heavy and final-heavy levels.
\end{lemma} 
\begin{proof}
We perform the charging level by level, starting from the lowest level $-1$. Given a level $\ell$, if it is either original-heavy or final-heavy then we do nothing. Otherwise, we match each epoch $\epoch\in \epoch_\ell$ that is either original or final with some distinct induced epoch $\epoch'\in \epoch_\ell$. We next charge the cost of $\epoch$ (as obtained from the proof of Lemma \ref{lem:inducedToOriginal}) to $\epoch'$. Finally, we charge the cost of $\epoch'$ to some original or final epoch $\epoch''$ at a higher level following the same scheme as in the proof of Lemma \ref{lem:inducedToOriginal}. Observe that at the end of this process only original and final epochs at the original-heavy and final-heavy levels are charged. Also, an easy induction shows that, when we start processing level $\ell$, the total charge on an original or final epoch at level $\ell$ coming from the lower levels is at most $\sum_{\ell'<\ell}O(3^{\ell'})=O(3^{\ell-1})$. The claim then follows. 
\end{proof}

\begin{lemma}\label{lem:final:new}
Given any sequence of $t$ updates, the total cost of all the final-heavy levels is (deterministically) at most  $O(t)$.
\end{lemma} 
\begin{proof}
Note that at each final-heavy level at least $1/8$-fraction of the epochs are final. The proof now follows from Lemma~\ref{lem:final}.
\end{proof}

\begin{corollary}\label{cor:charge}
The total cost of the data structure, excluding the preprocessing time and a term $O(t)$, is $O(1)$ times the total cost of the original-heavy levels. 
\end{corollary}
\begin{proof}
It follows from the above discussion and Lemmas~\ref{lem:lightToHeavy},~\ref{lem:final:new}.
\end{proof}

It now remains to bound the total cost of the original-heavy levels. This is the heart  of our analysis and the topic of the next section.  

\subsection{Bounding the Cost of the Original-Heavy Levels}

Recall that at each original-heavy level, at least $3/8$-fraction of the epochs are original. Thus, using a simple charging scheme, the task of bounding the total cost of all the original-heavy levels reduces to bounding the total cost of all the original epochs in these levels. At this point,  it is tempting to use the following argument. By Invariant \ref{main:inv:recolor}, for each epoch $\epoch$, $\ell=\ell(\epoch)$, the corresponding color $\chi(\epoch)$ is chosen uniformly at random in a palette of size at least $3^{\ell(\epoch)}/2+1$. Therefore, if $\epoch$ is original, we expect to see at least $\Omega(3^{\ell})$ edge insertions having $v(\epoch)$ as one endpoint before one such insertion causes a conflict with $v(\epoch)$. This would imply an $O(1)$ amortized cost per edge insertion. The problem with this argument is that, conditioning on an epoch $\epoch$ being original, modifies a posteriori the distribution of colors taken at the beginning of $\epoch$. For example, the choice of certain colors might make more likely that the considered epoch is induced rather than original. To circumvent this issue, we need a more sophisticated argument that exploits the fact that we are considering original epochs in original-heavy levels only.

We define the \emph{duration} $\dur(\epoch)$ of an epoch $\epoch$, $v=v(\epoch)$, as the number of edge insertions of type $uv$ that happen during $\epoch$, plus possibly the final insertion that causes the termination of $\epoch$ (if $\epoch$ is original). We also define a critical notion of \emph{pseudo-duration} $\psdur(\epoch)$ of $\epoch$ as follows. Let $vu_1,\ldots,vu_q$ be the subsequence of insertions of edges incident to $v$ in the input sequence after the creation of $\epoch$. Let us focus on the color $\chi(u_i)$ of $u_i$ \emph{right before the creation} of $\epoch$. Consider the sequence of colors $\chi(u_1),\ldots,\chi(u_q)$. Remove from this sequence all colors not in the palette $C$ used by $\epoch$ to sample $\chi(\epoch)$, and then leave only the first occurrence of each duplicated color. Let $\chi(1),\ldots,\chi(k)$ be the obtained subsequence of (distinct) colors.

We assume that $\chi(1),\ldots,\chi(k)$ is a permutation of $C$ (so that $k=|C|$), and otherwise extend it arbitrarily to enforce this property. We now define $\psdur(\epoch)$ to be the index $i$ such that $\chi(i)=\chi(\epoch)$.
\begin{lemma}\label{lem:domination}
For an original epoch $\epoch$, $\psdur(\epoch)\leq \dur(\epoch)$ deterministically.
\end{lemma}
\begin{proof}
Let $vu_i$ be the edge insertion that causes the termination of $\epoch$, so that $\dur(\epoch)=i$. Let $j\leq i$ be the smallest index with $\chi(u_j)=\chi(u_i)$. By definition, the value of $\psdur(\epoch)$ is the number of distinct colors in the set $\chi(u_1),\ldots,\chi(u_j)$. The latter number is clearly at most $j\leq i$.

Note that in this proof we crucially used the following aspect of our algorithm: If the insertion of an edge $(x, y)$ creates a conflict, in the sense that both $x$ and $y$ have the same color, then our algorithm changes the color of the node $z \in \{x, y\}$ that was {\em recolored last}.
\end{proof}

We say that an epoch $\epoch$ is \emph{short} if $\psdur(\epoch)\leq \frac{1}{32e}3^{\ell(\epoch)}$, and \emph{long} otherwise. The following critical technical lemma upper bounds the probability that an epoch is short. 
\begin{lemma}\label{lem:probShort}
An epoch $\epoch$ is short with probability at most $\frac{1}{16e}$, independently from the random bits used by the algorithm other than the ones used to sample $\chi(\epoch)$.
\end{lemma}
\begin{proof}
Let $C$ be the palette from which $v=v(\epoch)$ took its color $c=\chi(\epoch)$ uniformly at random. Let us condition on all the random bits used by the algorithm prior to the ones used to sample $\chi(\epoch)$. Notice that this fixes $C$ and the permutation $\chi(1),\ldots,\chi(|C|)$ of $C$ used for the definition of $\psdur(\epoch)$ (see the paragraph before Lemma \ref{lem:domination}). The random bits used after the sampling of $\chi(\epoch)$ clearly do not affect $\psdur(\epoch)$. The probability that $\psdur(\epoch)=1$, i.e. $\chi(i)=\chi(\epoch)$, is precisely $1/|C|$. The latter probability is deterministically at most $\frac{2}{3^{\ell(\epoch)}}$ by Invariant \ref{main:inv:recolor}. In particular, this upper bound holds independently from the random bits on which we conditioned earlier. The claim then follows since
$$
\Prob[\epoch \text{ is short}]=\Prob\left[\psdur(\epoch)\leq \frac{3^{\ell(\epoch)}}{32e}\right]=\frac{3^{\ell(\epoch)}}{32e}\cdot\frac{1}{|C|}\leq \frac{1}{16e}.
$$
\end{proof}

We next define some \emph{bad} events, that happen with very small probability. Given that those events do not happen, we can provide a good upper bound on the cost of heavy epochs. Recall that $\epoch_\ell$ is the set of epochs at level $\ell$.  We define $\epoch_\ell^{short}$ (resp., $\epoch_\ell^{long}$) as the collection of all epochs in $\epoch_\ell$ that are short (resp., long).
\begin{lemma}\label{lem:short-prob}
Consider any $x \geq 0$, and let $\al{\ell}{x}$ be the event that $|\epoch_\ell | > x$ and $|\epoch_\ell^{short}| \geq \frac{|\epoch_\ell |}{4}$. Then $\Prob(\al{\ell}{x}) \leq \frac{8}{3(2^x)}$. 
\end{lemma} 
\begin{proof}
Fix two parameters $q$ and $j$, with $j \ge q/4$, and consider any $q$ level-$\ell$ epochs $\cE^1, \dots , \cE^q$, ordered by their creation time. We argue that the probability that precisely $j$ particular epochs $\cE^{(1)}, \dots , \cE^{(j)}$ among these $q$ are short is at most $\left(\frac{1}{16e}\right)^j$. Let $B^{(i)}$ be the event that $\cE^{(i)}$ is short, $1 \leq i \leq j$. By a simple induction and Lemma~\ref{lem:probShort}, we have that  $\Prob(B^{(i)} ~\vert~ B^{(1)} \cap B^{(2)} \cap \ldots B^{(i-1)}) \leq \frac{1}{16e}$.
Consequently, $$\Prob(B^{(1)} \cap B^{(2)} \cap \ldots \cap B^{(j)}) ~=~ \Prob(B^{(1)}) \cdot \Prob(B^{(2)} ~\vert~ B^{(1)}) \cdot \ldots \cdot \Prob(B^{(j)} ~\vert~ B^{(1)} \cap B^{(2)}
\cap \ldots \cap B^{(j-1)}) ~\leq~ \left(\frac{1}{16e}\right)^j.$$

There are ${q \choose j}$ choices for the subsequence $\cE^{(1)}, \dots , \cE^{(j)}$, thus 
$$
\Prob[|\epoch_\ell | = q \cap |\epoch_\ell^{short} | = j] \leq {q \choose j} \left(\frac{1}{16e}\right)^j.
$$

Since ${q \choose j} \le (\frac{eq}{j})^j \le (4e)^j$, we can simplify ${q \choose j} \left(\frac{1}{16e}\right)^j \leq \frac{1}{4^j}$.
Hence,
\begin{align*}
\Prob(\al{\ell}{x}) 
= \sum_{q > x} \sum_{j = q/4}^q \Prob[|\epoch_\ell | = q \cap |\epoch_\ell^{short} | = j]
\leq \sum_{q > x} \sum_{j = q/4}^q  \frac{1}{4^j} 
\leq
\sum_{q >x} \frac{4}{3} \cdot \frac{1}{2^q}
\leq \frac{8}{3(2^x)}.
\end{align*} 
\end{proof}
\begin{corollary}\label{cor:probA} 
For a large enough constant $a>0$ and $x = a \log_{2} n$, let $A$ denote the event that  $\al{\ell}{x}$ happens for  {\em some} level $\ell$. Then $\Prob(A)=O(\frac{\log n}{n^a})$.
\end{corollary}
\begin{proof}
It follows from Lemma \ref{lem:short-prob} and the union bound over all levels $\ell$.
\end{proof}

\begin{lemma}\label{lem:durations-epochs}
The number $q$ of level-$\ell$ epochs with duration at least $\delta$ is bounded by $q \leq 2|IN_\ell|/\delta$ where $IN_\ell$ is the set of input insertions of edges incident to vertices at level $\ell$.
\end{lemma}
\begin{proof}
Observe that for the duration of the considered epochs, we consider only insertions in $IN_\ell$. Furthermore, each such insertion can influence the duration of at most $2$ such epochs. The claim follows by pigeon-holing. 
\end{proof}
 
 Let $c(\epoch_\ell)=O(3^\ell \cdot |\epoch_\ell|)$ denote the total cost of the epochs in level $\ell$.  We next relate the occurrence of event $\neg \al{\ell}{x}$ to the value of the random variable $c(\epoch_\ell)$ for original epochs.
\begin{lemma}\label{lem:terminated}
If $\neg\al{\ell}{x}$ occurs and level $\ell$ is original-heavy, then $c(\epoch_\ell)=O(|IN_\ell| + 3^{\ell} x)$.
\end{lemma}
\begin{proof}
If $|\epoch_\ell | \leq x$, then we clearly have $c(\epoch_\ell)=O(3^{\ell} x)$. For the rest of the proof, we assume that $|\epoch_\ell^{short}|< \frac{|\epoch_\ell |}{4}$, or equivalently: 
\begin{equation}
\label{eq:revise:new:1}
\left|\epoch_\ell^{long}\right| \geq \frac{3}{4} \cdot |\epoch_\ell|.
\end{equation}
Let $\epoch^*_\ell \subseteq \epoch_\ell$ be the set of original epochs at level $\ell$. Since the level $\ell$ is original-heavy, we have:
\begin{equation}
\label{eq:revise:new:2}
\left| \epoch^*_\ell \right| \geq \frac{3}{8} \cdot |\epoch_\ell|.
\end{equation}
Applying the pigeon-hole principle, from~(\ref{eq:revise:new:1}) and~(\ref{eq:revise:new:2}) we infer that at least $q\geq \frac{|\epoch_\ell|}{8}$ level-$\ell$ epochs are original and long at the same time. Specifically, we get:
\begin{equation}
\label{eq:revise:new:3}
\left| \epoch^*_\ell \cap \epoch_\ell^{long}\right| \geq \frac{1}{8} \cdot |\epoch_\ell|.
\end{equation}
Any such epoch $\epoch \in \epoch^*_\ell \cap \epoch_\ell^{long}$  has duration $\dur(\epoch)\geq \psdur(\epoch)\geq \frac{3^\ell}{32e}$ by Lemma \ref{lem:domination} and the definition of long epochs. Hence by applying Lemma \ref{lem:durations-epochs} with $\delta=\frac{3^\ell}{32e}$, we can conclude that the number of such epochs is at most $\frac{2|IN_\ell|}{3^\ell/(32e)}=\frac{64e|IN_\ell|}{3^\ell}$. Combining this observation with~(\ref{eq:revise:new:2}) we get:
\begin{equation}
\label{eq:revise:new:4}
|\epoch_\ell| \leq 8 \cdot \left| \epoch^*_\ell \cap \epoch_\ell^{long}\right| \leq 8 \cdot \frac{64e}{3^\ell} \cdot |IN_\ell|.
\end{equation}
 The claim follows if we multiply both sides of~(\ref{eq:revise:new:4}) by the $O(3^\ell)$ cost charged to each epoch in $\epoch_\ell$.
\end{proof}

\subsection{Bounding the Amortized Update Time of Our Data Structure}

We are now ready to prove that the amortized cost of our data structure is $O(1)$ with high probability (in $n$) for sufficiently long input sequences. Recall that $t$ denotes the total number of updates.
\begin{lemma}\label{lem:mainHighProbability}
With probability $1-O(\log n/n^{a})$, the total running time of our data structure over any sequence of $t$ updates is $O(t+n\log n+\Delta n)$.
\end{lemma}
\begin{proof}
By Lemma \ref{lem:runningTime}-\ref{lem:runningTime:prep} the preprocessing time is $O(n\log n+\Delta n)$. The total cost of deletion and conflict-less insertions is $O(t)$, due to  Lemma \ref{lem:runningTime}-\ref{lem:runningTime:easy}. Let us condition on the event $\neg A$, which happens with probability $1-O(\log n/n^{a})$ by Corollary \ref{cor:probA}. Then the total cost of the original-heavy levels is $O(\sum_\ell (|IN_\ell| + 3^{\ell} a\log n))=O(t+n\log n)$ by Lemma \ref{lem:terminated}.  The lemma now follows from Corollary \ref{cor:charge}.
\end{proof}

In order to prove that the amortized update time of our data structure is $O(1)$ in expectation, we also need the following upper bound on its worst-case running time. 
\begin{lemma}\label{lem:worstCase}
The running time of our data structure is deterministically at most $O(t n^2+n\log n+\Delta n)$.
\end{lemma}
\begin{proof}
Again, the preprocessing time is $O(n\log n+\Delta n)$ and the total cost of deletions and conflict-less insertions is $O(t)$ deterministically. Each conflicting insertion starts a sequence of calls to $\recolor(\cdot)$ involving some nodes $w_1,\ldots,w_q$. Notice that a given node $w$ can appear multiple times in the latter sequence. However, the sequence ends when some node $w^*$ is moved to level $-1$, and in all other cases the level of $w$ is increased by at least one. This means that the total cost associated with node $w$ is $O(\sum_{\ell}3^\ell)=O(n)$. The claim follows by summing over the $n$ nodes and the $O(t)$ insertions. 
\end{proof}
Hence we can conclude:
\begin{lemma}\label{lem:mainExpected}
The total expected running time of the above data structure is $O(t+n\log n+\Delta n)$.
\end{lemma}
\begin{proof}
When the event $\neg A$ happens, the total cost of the data structure is $O(t+n\log n+\Delta n)$ by Lemma \ref{lem:mainHighProbability}. If instead the event $A$ happens, then the cost is at most $O(tn^2+n\log n+\Delta n)$ by Lemma \ref{lem:worstCase}. However the latter event happens with probability at most $O(\frac{\log n}{n^a})$ by Corollary \ref{cor:probA}. Hence this second case adds $o(t)$ to the total expected cost for $a\geq 3$.
\end{proof} 

We now have all the ingredients to prove the main theorem of this paper.
\begin{proof}[Proof of Theorem \ref{thr:main}]
Consider the data structure described above. The space usage follows from Lemma \ref{lem:runningTime}-\ref{lem:runningTime:space} and the update time from Lemmas \ref{lem:mainHighProbability} and \ref{lem:mainExpected}.
\end{proof}

\bibliographystyle{abbrv}
\bibliography{citations}

\appendix

\input{data-structures}

\end{document}

%% file: data-structures.tex
\newcommand{\nhbr}{\mathcal{N}}

\section{$(\Delta + 1)$-Coloring Update Data Structures}
In this section, we give a full detailed description of the data structures used by our dynamic algorithm.

\subsection{Proof of Lemma~\ref{lem:runningTime}}
\begin{proof}
We now justify the  five claims made in the statement of the lemma. A full, detailed implementation section of the data structures can be found in Appendix~\ref{sec:update}.
\begin{enumerate}
\item We initialize an array $\up{v}$ for each vertex $v$ that contains $O(\log n)$ entries that stores the up-neighbors of each vertex. Each index of the array contains a pointer to a linked list containing the up-neighbors of the vertex at that level $\ell(v)$. We initialize another linked list $\down{v}$ which contains the down-neighbors of $v$. Furthermore, we initialize two linked lists, $\C^+_v$ and $\colorarr{v}$. $\C^+_v$ contains the colors of the up-neighbors stored in $\up{v}$. $\colorarr{v}$ contains the colors of the down-neighbors stored in $\down{v}$ and the blank colors: $\C \backslash \C^{+}_v$. The palette $\C^*_v$ containing the unique and blank colors of $v$ can thus be computed from $\colorarr{v}$. Each $\up{v}$ has size $O(\log n)$; $\C^+_v$ and $\colorarr{v}$ each has size $O(\Delta)$; and $\down{v}$ is initially empty. Thus, the preprocessing time necessary to initialize these structures is $O(n\log n+\Delta n)$. More details on these structures can be found in Section~\ref{sec:data-structures}.
\item The total space used by $\down{v}$ for all $v$ is $O(m)$ since $\down{v}$ for vertex $v$ stores at most the number of neighbors of $v$. All other data structures are initialized during preprocessing. Therefore, the space cost of the other data structures is $O(n\log n+\Delta n)$. We maintain mutual pointers between all versions of vertices in the various data structures.
\item Deleting an edge $uv$ requires deleting $u$ from $\up{v}$ and $v$ from $\down{u}$ (or vice versa). Inserting an edge $uv$ requires inserting $u$ into $\up{v}$ and $v$ into $\down{u}$ (or vice versa). The colors for $u$ and $v$ can be moved in between $\C^+_v$ and $\colorarr{v}$ and between $\C^+_u$ and $\colorarr{u}$ via a set of pointers connecting the colors to the vertices. Refer to Fig.~\ref{fig:handleIns},~\ref{fig:handleDel} and Section~\ref{sec:update} for a detailed description of these elementary operations. The total cost of these operations is then $O(1)$.
\item In this procedure $\detc(v)$, the level of $v$ is set deterministically to $-1$ and the color for $v$ is chosen deterministically from its set of blank colors. In this case, all the data structures of vertices in levels $[\ell^*(v), \ell(v)]$ must be updated with the new level of $v$. Due to the existence of pointers in between vertices and its neighbors in $\down{v}$ and $\up{v}$ in all the data structures, the cost of updating each individual neighbor is $O(1)$. To update the colors of the data structures requires following $O(1)$ pointers for each $w \in \down{v}$.
By the definition of $\recolor(v)$ (which calls $\detc(v)$), $\lbelow(\ell(v)+ 1) < \base^{\ell(v) + 2}$. Hence, there are $O(\base^{\ell(v)})$ neighbors in levels $[\ell'(v), \ell(v)]$ to update and the cost of the procedure is $O(\base^{\ell(v)})$. Refer to Section~\ref{sec:update} and Fig.~\ref{fig:detr} for a complete description of this procedure.
\item Since $\ell'(v) > \ell(v)$, all the data structures of vertices in levels $[\ell(v), \ell'(v)]$ must be updated with the new level of $v$. The data structures can be updated in the same way as given above. Since $\ell'(v) > -1$ (it must be, by definition of the procedure), then, $\lbelow(\ell'(v) + 1) < \base^{\ell'(v) + 2}$. Hence, this procedure takes $O(\base^{\ell'(v)})$ time. Refer to Section~\ref{sec:update} and Fig.~\ref{fig:rand} for a complete description of this procedure.
\end{enumerate}
\end{proof}

\subsection{Full Implementation Details}\label{sec:update}
The update algorithm is applied following edge insertions and deletions to and from the graph. In this section, we provide a complete description of the update data structures and algorithm. The pseudocode of this algorithm can be found in Appendix~\ref{app:pseudocode}. We begin with a description of the data structures and invariants that will be maintained by our algorithm.

\subsection{Hierarchical Partitioning and Coloring Data Structures}\label{sec:data-structures}
Our algorithm maintains the following set of data structures which we divide into two groups: the data structures responsible for maintaining our hierarchical partitioning and the data structures used to maintain the set of colors associated with each vertex. Let $\mathcal{C}$ be the set of all $\Delta + 1$ colors. The first group of data structures is a hierarchical partitioning of the vertices of the graph into different \emph{levels} according to some procedures that maintain a set of invariants. A vertex at a level have some number of neighbors in other levels of the hierarchical partitioning structure. We refer to neighbors at the same or higher levels of the hierarchical partitioning structure as the \emph{up-neighbors}. We refer to neighbors at lower levels of the hierarchical partitioning as the \emph{down-neighbors}. Different data structures will be used to maintain the colors of the down-neighbors and the colors of the up-neighbors of a vertex.

The second group of data structures deals with maintaining the colors of the vertices, inspired by the structures given in~\cite{BCHN18}. For the following data structures, we use logarithms in base $\base$ unless stated otherwise.

Let $\mathcal{C}$ be the set of all $\Delta + 1$ colors:

\begin{itemize}
    \item \textbf{Hierarchical Partitioning:} We maintain the following data structures necessary for our hierarchical partitioning. 
    \begin{enumerate}
    \item For each vertex $v$:
    \begin{enumerate}
            \item $\mathcal{N}_v$: a linked list containing all neighbors of $v$.
            \item $\down{v}$: a linked list containing all down-neighbors of $v$.
            \item $\up{v}$: a dynamic hash table where each index corresponds to a distinct level $\ell \in \{0, \dots, \log_{\base}(n-1) - 1\}$. $\up{v}[\ell]$ holds a pointer to the head of a non-empty doubly linked list containing all up-neighbors of $v$ at level $\ell$. If this list is empty, then the corresponding pointer is not stored. 
    \end{enumerate}
    \item For any vertex $v$ and any neighbor $u$ in $\down{v}$, we maintain mutual pointers between all elements $u \in \down{v}$, $v \in \up{u}[\ell(v)], v \in \mathcal{N}_u, u\in \mathcal{N}_v$. This means that for any neighbor, $u$ of $v$, we maintain mutual pointers $u \in \mathcal{N}_v$ and $u \in \up{v}$ \emph{or} $u \in \mathcal{N}_v$ and $u \in \down{v}$. This also means that given an edge insertion or deletion, we are able to quickly access the endpoints of the edge in each data structure once we locate one copy of an endpoint in memory.
    \item We define $\phi_v(\ell)$ to be the number of neighbors of $v$ with level strictly lower than $\ell$.
    \end{enumerate}
   \item \textbf{Coloring:} We maintain the following data structures for our coloring procedures. These structures are similar to the structures used in~\cite{BCHN18}.
        \begin{enumerate}
            \item A static array $\chi$ of size $O(n)$ where $\chi[i]$ stores the \emph{current} color of the $i$-th vertex.
            \item For each vertex $v$:
            \begin{enumerate}
                \item $\mathcal{C}^+_v$: a doubly linked list of colors occupied by vertices in $\up{v}$. Each color contains a counter $\mu_v^+(c)$ counting the number of vertices in $\up{v}$ that is colored color $c$.
                \item The counters $\mu_v^+(c)$ are stored in a static array of size $\Delta + 1$ where index $i$ contains the number of vertices in $\up{v}$ that is colored with color $i$.
                \item $\colorarr{v}$: a doubly linked list of colors in $\mathcal{C} \backslash \mathcal{C}^+_v$.
                \item A static array $\mathcal{P}_v$ of size $\Delta + 1$ containing mutual pointers (i.e. the pair of pointers from element $a$ to element $b$ and from element $b$ to element $a$) to each color $c$ in $\colorarr{v}$ or $\cp{v}$ and to each of two additional nodes representing each color in $\mathcal{C}$. Let $i_c$ be the index of color $c$ in $\mathcal{P}_v$. Suppose that $c \in \colorarr{v}$. Let $\po{c}$ and $\pp{c}$ be the two additional nodes representing $c$. Then $\mathcal{P}_v[i_c]$ contains pointers to $c \in \colorarr{v}$, $\po{c}$, and $\pp{c}$. In addition, if $c \in \colorarr{v}$, then it has mutual pointers to $\po{c}$. If, instead, $c \in \cp{v}$, then it has mutual pointers to $\pp{c}$ instead. In other words, $\po{c}$ receives pointers from nodes in $\colorarr{v}$ (and has outgoing pointers to nodes in $\colorarr{v}$) and $\pp{c}$ receives pointers from nodes in $\cp{v}$ (and has outgoing pointers to nodes in $\cp{v}$).
            \end{enumerate}
        \end{enumerate}
    \end{itemize}

We define the set of \emph{blank colors} for $v$ to be colors in $\colorarr{v}$ which are not occupied by any vertex in $\down{v}$. We define the set of \emph{unique colors} of $v$ to be colors in $\colorarr{v}$ which are occupied by at most one vertex in $\down{v}$.

   We now describe the pointers from the hierarchical partitioning structures to the coloring structures and vice versa.
    \begin{itemize}
        \item Each color $c$ in $\cp{v}$ has a pointer to $\pp{c}$ and vice versa.
        \item Each color $c'$ in $\colorarr{v}$ has a pointer to $\po{c'}$ and vice versa.
        \item Each vertex $u \in \up{v}$ contains mutual pointers to the node $\pp{c}$ representing its color in $\cP_v$ that it is currently colored with. The color $c$ is also in $\mathcal{C}_v^+$ and has mutual pointers to $\pp{c}$.
        \item Each vertex $u \in \down{v}$ contains mutual pointers to $\po{c}$ representing its corresponding color in $\cP_v$. If its color is in $\colorarr{v}$, then mutual pointers also exist between $\po{c}$ and $c$ in $\colorarr{v}$.
        \item Each edge $(u, v)$ contains two pointers, one to $u \in \mathcal{N}_v$ and one to $v \in \mathcal{N}_u$. $u$ and $v$ also contain pointers to edge $(u, v)$.
    \end{itemize}

\paragraph{Initial Data Structure Configuration, Time Cost, and Space Usage} There exist no edges in the graph initially; thus all vertices can be colored the same color. Such an arbitrary starting color is chosen. Before any edge updates are made, we assume that all vertices are on level $-1$, colored with the arbitrary starting color. Thus, all colors are also initially in $\colorarr{v}$.

Before any edge insertions, the only structures that we initialize are an empty hash table $\up{v}$ for each vertex $v$, the list of all colors $\mathcal{C}$, and $\chi$. When the first edge that contains vertex $v$ as an endpoint is inserted, we initialize $\cN_v$, $\down{v}$, $\up{v}$, $\cp{v}$, $\colorarr{v}$, $\mu_v^{+}$ for $v$ (as well as the associated pointers). The time for initializing these structures is $O(n\Delta)$ which means that the preprocessing time will result in $O(1)$ amortized time per update assuming $\Omega(n\Delta)$ updates.

We note a particular choice in constructing our data structures. In the case of $\up{v}$, given our assumption of the number of updates, we can also
implement $\up{v}$ as a static array instead of a dynamic hash table. The maximum number of levels is bounded by $\log_{\base} (n-1) + 1$.
Thus, if we instead implemented $\up{v}$ as static arrays instead of dynamic hash tables, the total space usage (and initialization cost) would be $O(n \log{n})$, amortizing to $O(1)$ per update given $\Omega(n\log n)$ updates. There may be reasons to implement $\up{v}$ as static arrays instead of dynamic hash tables such as easier implementation of basic functions. However, we choose to use a dynamic hash table implementation for potential future work for the cases when the number of updates is $O(n\log{n} + n\Delta)$. The key property we can potentially take advantage of in using the dynamic hash table implementation is that the total space used (and the total time spent in initializing the data structure) is within a constant factor of the number of edges in the graph at any particular time.

\paragraph{Usefulness of the Pointers} Pointers between the various data structures used for the hierarchical partitioning and for maintaining the coloring allows for us to quickly update the state following an edge insertion or deletion. For example, when an edge $uv$ is inserted or deleted,
we get pointers to $v \in \nhbr_u$ and $u \in \nhbr_v$, and through these pointers we delete all elements $u \in \down{v}, v \in \up{u}[\ell(v)],v \in \nhbr_u,u \in \nhbr_v$ and potentially
move a color from $\mathcal{C}_u^+$ to $\mathcal{C}_u$. The exact procedure for handling edge deletions is described later.

\subsection{Invariants}\label{sec:invariants}

Our update algorithm and data structures maintain the following invariant.

\begin{invariant}\label{inv:recolor}
The following hold for all vertices:
    \begin{enumerate}
        \item A vertex in level $\ell$ was last colored using a palette of size at least $\left(1/2\right)\base^{\ell + 1} + 1$.  As a special case, a vertex in level $-1$ was last colored using a palette of size $1$ (in other words, it was colored deterministically).
        \item The level of a vertex remains unchanged until the vertex is recolored.
    \end{enumerate}
\end{invariant}

\subsection{Edge Update Algorithm}\label{sec:edge-update}
We now describe the update algorithm in detail. The data structures are initialized as described in Section~\ref{sec:data-structures}. Then, edge updates are applied to the graph. Following an edge insertion or deletion, the procedure $\ins(u, v)$ or $\del(u, v)$, respectively, is called. The descriptions of the insertion and deletion procedures are given below.

\paragraph{Procedures $\ins(u,v)$ and $\del(u,v)$.}
Procedure $\ins(u,v)$ is called on an edge insertion $\ins(u, v)$. The pseudocode for this procedure is given in Figure~\ref{fig:handleIns}.
If edge $uv$ does not connect two vertices that are colored the same color (i.e.\ if the insertion is \emph{conflict-less}), then we only need to update the relevant data structures with the inserted edge. Namely the vertices are added to the structures maintaining the neighbors of $u$ and $v$. If $u$ is on a higher level than $v$, then $u$ is added to $\up{v}$ and $v$ is added to $\down{u}$ (and vice versa). If $u$ and $v$ are on the same level, then $u$ is added to $\up{v}$ and $v$ is added to $\up{u}$.
Furthermore, the colors that are associated with the vertices are moved in between the lists $\cp{v}$ and $\colorarr{v}$ as necessary. See
the pseudocode in Fig.~\ref{fig:handleIns} for exact details of these straightforward data structure updates.

In the case that edge $uv$ connects two vertices of the
same color (i.e.\ if the insertion is \emph{conflicting}), we need to recolor at least one of these two vertices. We arbitrarily recolor \emph{one} of the vertices $u$ or $v$ using procedure $\recolor$ (i.e.\ $\recolor(u)$ as given in the pseudocode in Fig.~\ref{fig:recolor}).
Procedure $\recolor$ is the crux of the update algorithm and is described next.

Procedure $\del(u,v)$ is called on an edge deletion $uv$. This case would not result in any need to recolor any vertices since a conflict will never be created.
Thus, we update the relevant data structures in the obvious way as stated above and as given in the pseudocode in Fig.~\ref{fig:handleDel}.
\\\\
Whenever a conflict is created following an edge insertion $uv$, procedure $\recolor(u)$ is called on one of the two endpoints. This procedure is described below.

\paragraph{Procedure $\recolor(v)$.}
The pseudocode for this procedure can be found in Figure~\ref{fig:recolor}. The procedure $\recolor(v)$ makes use of the level of $v$ as well as the number of its down-neighbors to either choose a blank color deterministically to recolor $v$ or to determine the palette from which to select a random color to recolor $v$. Recall that all vertices start in level $-1$ before any edges are inserted into the graph.

The procedure $\recolor(v)$ considers two cases:

\begin{itemize}
\item \emph{Case 1:} $\phi_v(\ell(v) + 1) < \base^{\ell(v)+2}$.
In other words, the first case is when the number of down-neighbors and vertices on the same level as $v$ is not much greater than $\base^{\ell(v)+1}$. We show in the analysis that in this case, we can find the colors of all the neighbors in $\down{v}$ and pick a color in $\colorarr{v}$ that does not conflict with any such neighbors (or the color that it currently has). Thus, we deterministically choose a blank color to recolor $v$, creating no further conflicts. The procedure to choose a blank color for $v$, $\detc(v)$, is described in the following.
\item \emph{Case 2:} $\phi(\ell(v) + 1) \ge \base^{\ell(v)+2}$. In this case, the number of down-neighbors and vertices on the same level as $v$ is at least $\base^{\ell(v)+2}$ and it will be too expensive to look for a blank color since we need to look at all neighbors in $\down{v}$ to determine such a color and the size of $\down{v}$ could be very large. Thus, we need to pick a random color from $\colorarr{v}$ to recolor $v$ by running Procedure $\randc(v)$ as described below.
\end{itemize}

\paragraph{Procedures $\detc(v)$ and $\randc(v)$.}
When called, the procedure\\
$\detc(v)$ starts by scanning the  list $\colorarr{v}$ to find at least one blank color that we can use to color $v$.
By the definition of $(\Delta + 1)$-coloring, there must exist at least one blank color with which we can use to color $v$. We can deterministically
find a blank color in the following way. The elements in $\colorarr{v}$ are stored in a doubly linked list. We start with the first element
at the front of the list and scan through the list until we reach an element that does not have a pointer to a vertex in $\down{v}$. We can determine whether a color $c \in \colorarr{v}$ has a pointer to a vertex in $\down{v}$ by following the pointer from $c$ to $\po{c}$. From $\po{c}$, we can then determine whether any vertices in $\down{v}$ are colored with $c$.

Let this first blank color be $c_b$. We assign
color $c_b$ to $v$, update $\chi(i_v)$ to indicate that the color of $v$ is $c_b$, and update the lists $\cp{w}$ and/or $\colorarr{w}$ of all $w \in \down{v}$. To update
all $\cp{w}$ and $\colorarr{w}$, we follow the following set of pointers:
\begin{enumerate}
    \item From $w \in \down{v}$, follow pointers to reach $w \in \cN_v$.
    \item From $w \in \cN_v$, follow pointers to reach $v \in \cN_w$.
    \item From $v \in \cN_w$, follow pointers to reach $v \in \up{w}[\ell(v)]$.
    \item Let $c$ be the previous color of $v$ as recorded in $\cp{w}$. From $v \in \up{w}[\ell(v)]$, follow pointers to reach $c \in \cp{w}$.
    \item Decrement $\cc{w}{c}$ by $1$. Delete mutual pointers between $v$ and $\pp{c}$. If now $\cc{w}{c} = 0$, remove $c$ from $\cp{w}$, append $c$ to the end of $\colorarr{w}$, delete mutual pointers between $c$ and $\pp{c}$, and add mutual pointers between $c$ and $\po{c}$.
    \item Use $\p{w}$ to find $c_b$ in either $\cp{w}$ or $\colorarr{w}$. If $c_b \in \cp{w}$, increment $\cc{w}{c_b}$ by $1$. Otherwise, if $c_b \in \colorarr{w}$, remove
    $c_b$ from $\colorarr{w}$, append $c_b$ to the end of $\cp{w}$, increment $\cc{w}{c_b}$ by $1$, delete mutual pointers between $c_b$ and $\po{c_b}$, and create mutual pointers between $c_b$ and $\pp{c_b}$. Create mutual pointers between $v \in \up{w}[\ell(v)]$ and $\pp{c_b}$.
\end{enumerate}

After the above is done in terms of recoloring the vertex $v$, $\setlvl(v, -1)$ is called to bring the level of $v$ down to $-1$. The description of $\setlvl(v, -1)$ is given in the following. See the pseudocode for $\detc(v)$ in Fig.~\ref{fig:detr} for concrete details of this procedure.
\\\\
The procedure $\randc(v)$ employs a \emph{level-rising mechanism}. We mentioned before the concept of partitioning vertices into levels.  Each level bounds the down-neighbbors of the vertices at that level, providing both an upper and lower bound on the number of down-neighbors of the vertex. Because there are at most $\log_3(n-1)$ levels, the number of vertices in each level is thus exponentially increasing. The procedure $\randc(v)$ takes advantage of this bound on the number of down-neighbors of the vertex $v$ to find a level to recolor $v$ with a color randomly chosen from its $\colorarr{v}$. Specifically, $\randc(v)$ recolors $v$ at some level $\ell^*$ higher than $\ell(v)$, with a random blank or unique color occupied by vertices of levels \emph{strictly lower} than $\ell^*$.
At level $\ell^*$, it attempts to select a color $c$ within time $O(\base^{\ell^*})$; this can only occur if $|\down{v}| = O(\base^{\ell^*})$.
Upon failure, it calls itself recursively to color $v$ at yet a higher level. Again, $\setlvl(v, \ell^*)$ is called every time $v$ moves to a high level.

\paragraph{Procedure $\setlvl(v,\ell)$.}
Procedures $\detc(v)$ and $\randc(v)$ may set the level of $v$ to a different level, in which case the procedure $\setlvl(v, \ell)$ is called with the new level $\ell$ as input.
Let $\ell(v)$ be the previous level of $v$. The procedure does nothing if $\ell = \ell(v)$. Otherwise:

\textit{If $v$ is set to a lower level $\ell < \ell(v)$: } we need to update the data structures of vertices in levels $[\ell+1,\ell(v)]$.
For each vertex $w \in \down{v}$ where $\ell+1\le \ell(w) \le \ell(v)$, we make the following data structure updates:
\begin{enumerate}
    \item Delete $w$ from $\down{v}$. Delete the mutual pointers between $w$ and $\po{c}$. Let $w$'s color be $c$. Move $w$'s color, $c$, in $\colorarr{v}$ to $\cp{v}$ if $c$ is currently in $\colorarr{v}$. Delete the mutual pointers between $c$ and $\po{c}$. Create mutual pointers between $c$ and $\pp{c}$. Increment $w$'s color count $\mu_v^+(c)$ by $1$.
    \item Add $w$ to  $\up{v}[\ell(w)]$. Add mutual pointers between $w$ and $\pp{c}$ where $c$ is $w$'s color.
    \item Delete $v$ from $\up{w}[\ell(v)]$. Let $v$'s color be $c'$. Delete the mutual pointers between $v$ and $\pp{c'}$. Decrement $v$'s color count $\mu_w^+(c')$ by $1$. If $\mu_w^+(c')$ is now $0$, move $c'$ from $\cp{w}$ to $\colorarr{w}$, delete the mutual pointers between $c'$ and $\pp{c'}$, and create mutual pointers between $c'$ and $\po{c'}$.
    \item Add $v$ to $\down{w}$. Add mutual pointers between $v$ and $\po{c'}$ where $c'$ is $v$'s color if $c'$ was moved to $\colorarr{w}$.
    \item Add mutual pointers between all elements $v \in \down{w}, w \in \up{v}[\ell(w)],v \in \cN_w,w \in \cN_v$.
    \item Add mutual pointers between all copies of the same element: i.e.\ $w \in \down{v}, w \in \cN_v,$ and/or $w \in \up{v}[\ell(w)]$.
    \item Maintain mutual pointers between $\mathcal{P}_v[i_c]$, $\po{c}$, and $\pp{c}$. Maintain mutual pointers between $\mathcal{P}_w[i_{c'}]$, $\pp{c'}$, and $\po{c}$.
\end{enumerate}

\textit{If $v$ is set to a higher level $\ell > \ell(v)$:} we need to update the data structures of vertices in levels $[\ell(v), \ell-1]$.
Specifically, for each non-empty list $\up{v}[i]$, with $\ell(v) \le i \le \ell-1$,
and for each vertex $w \in \up{v}[i]$, we perform the following operations:

    \begin{enumerate}
    \item Delete $w$ from $\up{v}[i]$. Let $c$ be the color of $w$. Delete the mutual pointers between $w$ and $\pp{c}$. Decrement $\cc{v}{c}$ by $1$. If $\cc{v}{c} = 0$, then move $c$ from $\cp{v}$ to $\colorarr{v}$, delete the mutual pointers between $\pp{c}$ and $c$, and add mutual pointers between $\po{c}$ and $c$.
    \item Add $w$ to $\down{v}$, create mutual pointers between $w$ and $\po{c}$ (where $c$ is $w$'s color), delete $v$ from $\down{w}$, and add $v$ to $\up{w}[\ell]$. Let $v$'s color be $c'$. Delete the mutual pointers between $v$ and $\po{c'}$. Add mutual pointers between $v$ and $\pp{c'}$. If $c'$ is currently in $\colorarr{w}$, move $v$'s color, $c'$, in $\colorarr{w}$ to $\cp{w}$, delete mutual pointers between $c'$ and $\po{c'}$, and add mutual pointers between $c'$ and $\pp{c'}$. Increment $\cc{w}{c'}$ by 1.
    \item Add mutual pointers between all elements $w \in \down{v}, v \in \up{w}[\ell],v \in N_w,w \in N_v$.
    \item Maintain mutual pointers between $\cP_v[i_{c}]$, $\po{c}$, and $\pp{c}$. Maintain mutual pointers between $\cP_w[i_{c'}]$, $\po{c'}$, and $\pp{c'}$.
\end{enumerate}

The full pseudocode of this procedure can be found in Fig.~\ref{fig:setlvl}.

\section{Pseudocode}\label{app:pseudocode}
In the below pseudocode, we do not describe (most of) the straightforward but tedious pointer creation procedures. We assume that the corresponding pointers
are created according to the procedure described in Section~\ref{sec:update}. In the cases where the pointer change is significant, we describe it in the pseudocode.

\begin{figure}[!ht]
\fbox{
        \begin{minipage}[t]{165mm}
$\ins(u,v)$:
\begin{enumerate}
\item $\nhbr_v \leftarrow \nhbr_v \cup \{u\}$;
\item $\nhbr_u \leftarrow \nhbr_u \cup \{v\}$;
\item If $\ell(u) > \ell(v)$:
   \begin{enumerate}
   \item $\down{u} \leftarrow \down{u} \cup \{v\}$;
   \item $\up{v}[\ell(u)] \leftarrow \up{v}[\ell(u)] \cup \{u\}$;
   \end{enumerate}
\item Else if $\ell(u) = \ell(v)$:
	\begin{enumerate}
		\item $\up{v}[\ell(u)] \leftarrow \up{v}[\ell(u)] \cup \{u\}$;
		\item $\up{u}[\ell(v)] \leftarrow \up{u}[\ell(v)] \cup \{v\}$;
	\end{enumerate}
\item Else:
   \begin{enumerate}
   \item $\down{v} \leftarrow \down{v} \cup \{u\}$;
   \item $\up{u}[\ell(v)] \leftarrow \up{u}[\ell(v)] \cup \{v\}$;
   \end{enumerate}
\item $\updateins(u, v, c_u, c_v)$;
\item If $\text{color}(u) = \text{color}(v)$:  /* if $u$ and $v$ have the same color */
   \begin{enumerate}
   \item  $\recolor(u)$;
   \end{enumerate}
\end{enumerate}
\end{minipage}}
\caption{Handling edge insertion $(u,v)$.
}
\label{fig:handleIns}
\end{figure}

\begin{figure}[!ht]
\fbox{
        \begin{minipage}[t]{165mm}
$\del(u,v)$:
\begin{enumerate}
\item $\nhbr_v \leftarrow \nhbr_v \setminus \{u\}$;
\item $\nhbr_u \leftarrow \nhbr_u \setminus \{v\}$;
\item If $v \in \down{u}$:
   \begin{enumerate}
   \item $\down{u} \leftarrow \down{u} \setminus \{v\}$;
   \item $\up{v}[\ell(u)] \leftarrow \up{v}[\ell(u)] \setminus \{u\}$;
   \end{enumerate}
\item Else if $u \in \down{v}$:
   \begin{enumerate}
   \item $\down{v} \leftarrow \down{v} \setminus \{u\}$;
   \item $\up{u}[\ell(v)] \leftarrow \up{u}[\ell(v)] \setminus \{v\}$;
   \end{enumerate}
\item Else:
	\begin{enumerate}
		\item $\up{u}[\ell(v)] \leftarrow \up{u}[\ell(v)] \setminus \{v\}$;
		\item $\up{v}[\ell(u)] \leftarrow \up{v}[\ell(u)] \setminus \{u\}$;
	\end{enumerate}
\item Remove all associate color pointers and shift colors between $\colorarr{u}$, $\cp{u}$ and $\colorarr{w}$, $\cp{w}$ as necessary;
\end{enumerate}
\end{minipage}}
\caption{Handling edge deletion $(u,v)$.
}
\label{fig:handleDel}
\end{figure}

\begin{figure}[!ht]
\fbox{
        \begin{minipage}[t]{165mm}
$\setlvl(v,\ell)$:
\begin{enumerate}
    \item For all $w \in \down{v}$:  /* update $\up{w}$ regarding $v$'s new level  */
       \begin{enumerate}
        \item $\up{w}[\ell(v)] \leftarrow \up{w}[\ell(v)] \setminus \{v\}$;
        \item $\up{w}[\ell] \leftarrow \up{w}[\ell] \cup \{v\}$;
        \end{enumerate}

\item If $\ell < \ell(v)$: /* in this case the level of $v$ is decreased by at least one */
    \begin{enumerate}
    \item For all $w \in \down{v}$ such that $\ell \leq \ell(w) < \ell(v)$:
     /* reassign color pointers*/
       \begin{enumerate}
        \item $\down{v} \leftarrow \down{v} \setminus \{w\}$;
        \item Delete mutual pointers between $w$ and $\po{\text{color}(w)}$;
        \item If $\text{color}(w) \in \colorarr{v}$:
        \begin{enumerate}
            \item Move $\text{color}(w) \in \colorarr{v}$ to $\cp{v}$;
            \item Delete mutual pointers between $\text{color}(w)$ and $\po{\text{color}(w)}$;
            \item Create mutual pointers between $\text{color}(w)$ and $\pp{\text{color}(w)}$;
        \end{enumerate}
        \item Increment $\cc{v}{\text{color}(w)}$ by $1$;
        \item $\up{v}[\ell(w)] \leftarrow \up{v}[\ell(w)] \cup \{w\}$;
        \item Create mutual pointers between $w$ and $\pp{\text{color}(w)}$ if such pointers do not already exist;
        \item $\up{w}[\ell] \leftarrow \up{w}[\ell] \setminus \{v\}$;
        \item Delete mutual pointers between $\text{color}(v)$ and $\pp{\text{color}(v)}$;
        \item Decrement $\cc{w}{\text{color}(v)}$ by $1$;
        \item If $\cc{w}{\text{color}(v)} = 0$:
        \begin{enumerate}
            \item Move $\text{color}(v) \in \cp{w}$ to $\colorarr{w}$;
            \item Delete mutual pointers between $\text{color}(v)$ and $\pp{\text{color}(v)}$;
            \item Create mutual pointers between $\text{color}(v)$ and $\po{\text{color}(v)}$;
        \end{enumerate}
        \item $\down{w} \leftarrow \down{w} \cup \{v\}$;
        \item Create mutual pointers between $v$ and $\po{\text{color}(v)}$ if such pointers do not already exist;
        \end{enumerate}
   \end{enumerate}

\item If $\ell > \ell(v)$: /* in this case the level of $v$ is increased by at least one */ \footnote{For the sake of clarity and brevity,
we do not describe the pointer deletions, creations, and changes in the case where $\ell > \ell(v)$ because these changes
are almost identical to the changes given above for the case $\ell < \ell(v)$.}
 \begin{enumerate}
    \item For all $i = \ell(v),\ldots,\ell-1$ and all $w \in \up{v}[i]$:
       \begin{enumerate}
       \item $\up{v}[i] \leftarrow \up{v}[i] \setminus \{w\}$;
       \item Decrement $\cc{v}{\text{color}(w)}$ by $1$;
       \item If $\cc{v}{\text{color}(w)} = 0$: Move $\text{color}(w)$ from $\cp{v}$ to $\colorarr{v}$;
       \item $\down{v} \leftarrow \down{v} \cup \{w\}$;
       \item $\down{w} \leftarrow \down{w} \setminus \{v\}$;
       \item $\up{w}[\ell] \leftarrow \up{w}[\ell] \cup \{v\}$;
       \item If $\text{color}(v) \in \colorarr{w}$: Move $\text{color}(v)$ from $\colorarr{w}$ to $\cp{w}$;
       \item Increment $\cc{w}{\text{color}(v)}$ by $1$;
       \end{enumerate}
\end{enumerate}
\item  $\ell(v) \leftarrow \ell$;
\end{enumerate}
\end{minipage}}
\caption{Setting the old level $\ell(v)$ of $v$ to $\ell$.}
\label{fig:setlvl}
\end{figure}

\begin{figure}[!ht]
\fbox{
        \begin{minipage}[t]{165mm}
$\recolor(v)$:
\begin{enumerate}
\item If $\phi_v(\ell(v) + 1) < \base^{\ell(v)+2}$: $\detc(v)$;
\item Else $\randc(v)$;
\end{enumerate}
\end{minipage}}
\caption{Recoloring a vertex that collides with the color of an adjacent vertex after an edge insertion.
}
\label{fig:recolor}
\end{figure}

\begin{figure}[!ht]
\fbox{
        \begin{minipage}[t]{165mm}
$\detc(v)$:
\begin{enumerate}
\item For all $c \in \colorarr{v}$:
       \begin{enumerate}
        \item If $c$ is not occupied by any vertex $w \in \down{v}$ and $c \in \colorarr{v}$: /* if $c$ is a blank color, color $v$ with $c$ */
        \begin{enumerate}
        \item Set $\chi\left(i_v\right) = c$;
        \item For all $w \in \down{v}$:
        \begin{enumerate}
            \item $\update(v, w, c)$.
        \end{enumerate}
        \item $\setlvl(v,-1)$;
        \item terminate the procedure; /* Note that the procedure will always terminate within this if statement because a blank color always exists by definition of $(\Delta + 1)$-coloring. */
        \end{enumerate}
   \end{enumerate}
\end{enumerate}
\end{minipage}}
\caption{Coloring $v$ deterministically with a blank color.
It is assumed that $\phi_v(\ell(v) + 1)  < \base^{\ell(v)+2}$.
}
\label{fig:detr}
\end{figure}

\begin{figure}[!ht]
\fbox{
        \begin{minipage}[t]{165mm}
$\randc(v)$:
\begin{enumerate}
\item $\ell^* \leftarrow \ell(v)$;
\item while $\phi_v(\ell^*+1) \ge \base^{\ell^*+2}$: $\ell^* \leftarrow \ell^* + 1$;
\\ /* $\ell^*$ is the minimum level after $\ell(v)$ with $\phi_v(\ell^*+1) < \base^{\ell^*+2}$ */
\item $\setlvl(v,\ell^*)$; /* after this call $\ell(v) = \ell^*$ and $\base^{\ell^* + 1} \le \dout(v) = \phi_v(\ell^*) < \base^{\ell^* + 2}$ */

\item Pick a blank or unique color $c$ from $\colorarr{v}$ uniformly at random;
\\ /* $c$ is chosen with probability at most $2/\base^{\ell^* + 1}$ and $\ell(w) \le \ell^* -1$ */
\item If $c \neq \text{color}(v)$: /* If $c$ is not the previous color of $v$. */
\begin{enumerate}
\item Set $\chi\left(i_v\right) = c$;
\item For all $z \in \down{v}$:
        \begin{enumerate}
            \item $\update(v, z, c)$.
        \end{enumerate}
\end{enumerate}
\item If $c$ is a unique color (let $w \in \phi_v\left(\ell^*\right)$ be the vertex that is colored with $c$):
\begin{enumerate}
    \item $\recolor(w)$;
\end{enumerate}
\end{enumerate}
\end{minipage}}
\caption{Coloring $v$ at level $\ell^*$ higher than $\ell(v)$, with a random blank or unique color of level lower than $\ell^*$.
If the procedure chose a unique color, it calls $\recolor$ (which may call itself recursively) to color $w$.
It is assumed that $\phi_v(\ell(v) + 1)  \ge \base^{\ell(v)+2}$.
}
\label{fig:rand}
\end{figure}

\begin{figure}[!ht]
\fbox{
        \begin{minipage}[t]{165mm}
$\updateins(v, w, c_v, c_w)$:
\begin{enumerate}
	\item If $\ell(v) >\ell(w)$:
	\begin{enumerate}
    \item Locate $v \in \up{w}[\ell(v)]$;
    \item Delete the mutual pointers (if they exist) between $v$ and $\pp{c'}$ where $c'$ is $v$'s previous color; /* Note that $v$'s previous color could be located by following pointers from $v$. */
    \item Decrement $\cc{w}{c'}$ by $1$ if pointers were deleted in the previous step; /* If no pointers were deleted, then $w$ had no knowledge of $v$'s previous color and we do not need to decrement */
    \item If $\cc{w}{c'} = 0$:
    \begin{enumerate}
        \item Move $c'$ from $\cp{w}$ to $\colorarr{w}$ by appending $c'$ to the end of the linked list representing $\colorarr{w}$;
    \end{enumerate}
    \item Locate $\pp{c_v}$ by following pointers from $\cP_w$;
    \item Create mutual pointers between $v$ and $\pp{c_v}$;
    \item Increment $\cc{w}{c_v}$ by $1$;
    \item If $c_v$ is in $\colorarr{w}$:
    \begin{enumerate}
        \item Move $c_v$ from $\colorarr{w}$ to $\cp{w}$ by appending $c_v$ to the end of the linked list representing $\cp{w}$;
    \end{enumerate}
    \item Locate $w \in \down{v}$. 
    \item Delete the mutual pointers (if they exist) between $w$ and $\po{c''}$ where $c''$ is $w$'s previous color;
    \item Locate $\po{c_v}$ by following pointers from $\cP_v$;
    \item Create mutual pointers between $w \in \down{v}$ and $\po{c_w}$;
\end{enumerate}
\item Else if $\ell(v) < \ell(w)$:
	\begin{enumerate}
		\item /* Do the above except switch the roles of $v$ and $w$ as well as $c_v$ and $c_w$. */
	\end{enumerate}
\item Else:
		\begin{enumerate}
			\item Locate $v \in \up{w}[\ell(v)]$ and $w \in \up{v}[\ell(w)]$;
			\item /* Do the above procedure given in the case when $\ell(v)> \ell(w)$ for $v \in \up{w}[\ell(v)]$ for both $v \in \up{w}[\ell(v)]$ and $w \in \up{v}[\ell(w)]$.*/
		\end{enumerate}
\end{enumerate}
\end{minipage}}
\caption{Updates the data structures with $v$ and $w$'s colors when an edge is inserted between $v$ and $w$.}
\label{fig:update-ins}
\end{figure}

\begin{figure}[!ht]
\fbox{
        \begin{minipage}[t]{165mm}
$\update(v, c_v)$:
\begin{enumerate}
	\item For $w \in \down{v}$:
	\begin{enumerate}
    \item Locate $v \in \up{w}[\ell(v)]$;
    \item Delete the mutual pointers (if they exist) between $v$ and $\pp{c'}$ where $c'$ is $v$'s previous color; /* Note that $v$'s previous color could be located by following pointers from $v$. */
    \item Decrement $\cc{w}{c'}$ by $1$ if pointers were deleted in the previous step; /* If no pointers were deleted, then $w$ had no knowledge of $v$'s previous color and we do not need to decrement */
    \item If $\cc{w}{c'} = 0$:
    \begin{enumerate}
        \item Move $c'$ from $\cp{w}$ to $\colorarr{w}$ by appending $c'$ to the end of the linked list representing $\colorarr{w}$;
    \end{enumerate}
    \item Locate $\pp{c_v}$ by following pointers from $\cP_w$;
    \item Create mutual pointers between $v$ and $\pp{c_v}$;
    \item Increment $\cc{w}{c_v}$ by $1$;
    \item If $c_v$ is in $\colorarr{w}$:
    \begin{enumerate}
        \item Move $c_v$ from $\colorarr{w}$ to $\cp{w}$ by appending $c_v$ to the end of the linked list representing $\cp{w}$;
    \end{enumerate}
   \end{enumerate}
   \end{enumerate}
\end{minipage}}
\caption{Updates the color pointers of $v$ of all of $v$'s down-neighbors when $v$ changes color.}
\label{fig:update}
\end{figure}